\documentclass[10pt,letterpaper]{article}
\usepackage[left=1 in, right=1 in, top=1 in, bottom = 1 in]{geometry}

\usepackage{hyperref}
\usepackage{microtype}
\usepackage{xcolor}
\usepackage{graphicx}
\usepackage{subfigure}
\usepackage{enumitem}
\usepackage{authblk}
\usepackage{amsmath, bm}
\usepackage{amssymb}
\usepackage{mathtools}
\usepackage{amsthm}
\usepackage{cleveref}
\usepackage{algorithm}
\usepackage{algorithmic}
\usepackage[numbers]{natbib}
\date{}

\newcommand{\D}{\mathcal{D}}

\newcommand{\N}{\mathcal{N}}
\newcommand{\J}{\mathcal{J}}
\newcommand{\predJ}{\hat{\mathcal{J}}}

\newcommand{\A}{\mathcal{A}}
\renewcommand{\S}{\mathcal{S}}
\renewcommand{\P}{\mathcal{P}}

\newcommand{\OPT}{\textsc{Opt}}

\DeclareMathOperator*{\argmax}{argmax}

\theoremstyle{plain}
\newtheorem*{theorem*}{Theorem}
\newtheorem{theorem}{Theorem}[section]

\newtheorem{lemma}[theorem]{Lemma}
\newtheorem{claim}[theorem]{Claim}

\theoremstyle{definition}
\newtheorem{definition}[theorem]{Definition}

\theoremstyle{remark}

\usepackage{thmtools}
\usepackage{thm-restate}

\newcommand{\tl}{t_{\lambda}}

\newcommand{\laps}{\textsc{LAP}}
\newcommand{\INP}{\item[{\bf Input:}]}

\title{Learning-Augmented Online Packet Scheduling with Deadlines}

\begin{document}
\author[1]{Ya-Chun Liang\thanks{ {\tt yl5608@columbia.edu}.}}
\author[1]{Clifford Stein\thanks{ {\tt cliff@ieor.columbia.edu}.}}
\author[1]{Hao-Ting Wei\thanks{ {\tt hw2738@columbia.edu}.}}
\affil[1]{Department of Industrial Engineering and Operations Research, Columbia University}

\maketitle

\begin{abstract}
The modern network aims to prioritize critical traffic over non-critical traffic and effectively manage traffic flow. This necessitates proper buffer management to prevent the loss of crucial traffic while minimizing the impact on non-critical traffic. Therefore, the algorithm's objective is to control which packets to transmit and which to discard at each step. In this study, we initiate the learning-augmented online packet scheduling with deadlines and provide a novel algorithmic framework to cope with the prediction. We show that when the prediction error is small, our algorithm improves the competitive ratio while still maintaining a bounded competitive ratio, regardless of the prediction error.
\end{abstract}

\section{Introduction}
In today's interconnected world, efficient and reliable data transmission is paramount. Businesses, governments, and individuals depend on seamless communication for various time-sensitive applications. Managing network traffic flow is critical to ensure that essential data takes precedence over non-critical information. One approach to achieving this is through the use of Quality of Service (QoS) switches~\cite{shin2001quality, zhang2004end, etoh2005advances}.

In the context of modern networks, the problem of Buffer Management in QoS switches arises as a significant challenge. The objective is to develop effective algorithms that prioritize critical traffic, safeguarding it from loss, while minimizing the impact on non-critical data. 
To emphasize the significance of this challenge, consider a data center \cite{ranjan2002qos, beloglazov2012managing} that processes incoming requests from users. Each request corresponds to a packet, each with a specific deadline; a time limit within which it must be transmitted to ensure a responsive user experience. The challenge at hand involves making real-time decisions on which packets to transmit while adhering to their respective deadlines. For instance, a delay in transmitting a critical update to a cloud-based application can lead to service disruptions, affecting user satisfaction and business operations. Similarly, 
timely delivery of multimedia content is essential for seamless streaming and online gaming experiences. A delay in transmitting video or audio packets can result in buffering or lag, affecting user engagement.

To achieve good performance, we design a scheduling algorithm, which decides  which packet should be processed at each step. As a consequence, some packets will miss their deadlines and need to be discarded. To ensure that the most critical traffic is processed, our objective is to maximize the total weight of processed jobs. To evaluate the performance of such algorithms, we use competitive analysis, a standard metric in online algorithms. The competitive ratio is used as a measure, representing the ratio of the optimal offline cost to the cost achieved by online algorithms in worst-case scenarios.

However, traditional competitive analysis can be overly pessimistic, as it considers only the worst-case scenario and overlooks the potential benefits of leveraging historical data to predict future inputs. To address this limitation, a new field of algorithms with predictions has emerged. These algorithms assume that the designer has access to predictions about the input, aiming to improve competitive ratios when predictions are accurate (consistency) while maintaining bounded competitive ratios even if the predictions are inaccurate (robustness).

In this paper, we study the packet scheduling problem with deadlines that we augment with predictions. This problem has not been previously studied with predictions. 
Much of the previous work on packet scheduling with deadlines has considered various special cases, which make assumptions about the input. These include \emph{agreeable deadlines}, where job's deadlines are (weakly) increasing in their arrival times; \emph{$k$-uniform delay}, where each job has the same $k$ amount of time available for processing; and \emph{$k$-bounded delay}, where each job has at most $k$ units of time available for processing, where $k$ is a positive integer. 
None of these cases have been studied with predictions.



\subsection{Our results} In an instance of  the \emph{packet scheduling with deadlines} problem, a collection $\J$ of $n$ jobs with unit processing time is given. Each job $(r_j, d_j, w_j) \in \J$ consists of a release time $r_j$, a nonnegative weight $w_j$, and an integer deadline $d_j$, indicating the time by which it must be processed. The goal is to find a schedule that maximizes the total weight of the processed jobs. In the \emph{learning-augmented packet scheduling with deadlines} problem, the algorithm is provided with a collection $\hat{\J}$ of $\hat{n}$ predicted jobs $(\hat{r}_j, \hat{d}_j, \hat{w}_j)$ at time $t=0$.

\noindent\textbf{Optimal consistency and bounded robustness.} Our first goal is to design an algorithm that strikes a good balance between consistency (competitive ratio when the predictions are exactly correct) and robustness (competitive ratio when the predictions are arbitrarily wrong). Our first result shows that for the learning-augmented packet scheduling with deadlines problem, given an offline prediction and an online algorithm with a competitive ratio of $\gamma_{\texttt{on}}$, there exists a learning-augmented algorithm that achieves 1 consistency (Lemma~\ref{lem:cons}) and 
$\gamma_{\texttt{on}} + c$ 
robustness (Lemma~\ref{lem:robust}), where $c$ is a constant. We complement this result by showing that there is a necessary trade-off between consistency and robustness (Section \ref{sec:lower-bound}).

\noindent\textbf{The competitive ratio as a function of the prediction error.} The second main result is that our algorithm achieves a competitive ratio that smoothly interpolates from consistency to 
robustness as a function of the prediction error (Lemma~\ref{lem:smooth}). 
We define the prediction error $\eta = \max_{t\in [T]} \frac{W(\OPT(\J_{\le t})^{(\le t)})}{W(\hat{\S}(\J)^{(\le t)})}$, which represents the maximum ratio between the total weight of the processed jobs in the optimal schedule for all released jobs until time $t$ and the processed weights obtained by following the choices of the optimal schedule for the prediction until time $t$, considering all possible values of $t$. Here, $\J_{\le t}$ and $\J^{(\le t)}$ denote the set of jobs in $\J$ with release time at most $t$ and the set of processed jobs in the interval $[0,t]$, respectively. (See Section~\ref{sec:prelim} for details of the discussion about the prediction error $\eta$).

\noindent\textbf{Experiments.} In Section~\ref{sec:exp}, we empirically assess the performance of our proposed algorithm by comparing it with different benchmark algorithms. We specifically investigate the impact of predictions of both weights and deadlines. Our results demonstrate that, when the prediction error is small, our algorithm outperforms the benchmark algorithms empirically on both real and synthetic datasets in the context of online packet scheduling with deadlines.

\subsection{Technical overview}
This paper proposes a general  algorithmic framework for learning-augmented online packet scheduling with deadlines. This framework relies on two critical components: an offline prediction and the online algorithm \textsc{OnlineAlg} without prediction given as input. For the offline setting of online packet scheduling with deadlines, one can solve a maximum weight bipartite matching problem to obtain the optimal solution \cite{kesselman2004buffer}. 
Various online algorithms are available for different variants of online packet scheduling with deadlines.  See   \cite{goldwasser2010survey} for a survey.

Initially, the algorithm computes an offline schedule $\hat{S}$ based on the prediction. While this schedule is optimal for perfect predictions, it may not be so in the presence of errors. To address this, the algorithm sets a threshold and conducts a local test to track the performance. 
At each time $t$,
the algorithm first checks whether $\hat{\S}(\J)^{(t)} \in \P $,
where $\hat{\S}(\J)^{(t)}$ is the processed job corresponding to $\hat{S}$
at time $t$, and $\P$ denotes the processed jobs of the algorithm until time $t-1$.  
If $\hat{\S}(\J)^{(t)}$ has not been processed, the algorithm calculates two quantities and derives the ratio between them. 
One is $W(\OPT(\J_{\le t})^{(\le t)})$, representing the total weight of the processed jobs in the optimal schedule for all released jobs until time $t$.
The other is the total weight of the processed jobs of the algorithm until time $t-1$, denoted by $W(\P)$, with the addition of $W(\hat{\S}(\J)^{(t)})$.

Either $\hat{\S}(\J)^{(t)} \in \P$, or if the ratio exceeds the threshold, the algorithm immediately switches to \textsc{OnlineAlg} until $\hat{\S}(\J)^{(t)} \notin \P$ and the local test ratio becomes no greater than the threshold. After that, the algorithm resumes following the choice of the predicted schedule.
This flexibility to seamlessly switch between offline and online decisions provides more opportunities to maintain superior performance when handling longer input sequences, setting our algorithm apart from many existing learning-augmented algorithms that strictly adhere to the online decisions after a transition~\cite{DBLP:journals/jacm/LykourisV21,antoniadis2021novel}.
 
Our analysis first shows that with an appropriate benchmark for the local test, we can ensure that when the prediction is correct, the algorithm will always follow the choice of the predicted schedule; that is, the output schedule is optimal. Additionally, we show that our proposed algorithm achieves constant robustness. 
The analysis involves partitioning the entire time interval into two smaller intervals, namely $[0, \tl-1]$ and $[\tl, T]$, 
where $\tl-1$ is the last time that the local test ratio is not greater than the threshold.
%
%
The high-level idea behind the analysis is as follows: suppose a job in the optimal solution is processed in the interval $[\tl, T]$. In this case, it is either processed by our proposed algorithm in the interval $[0,\tl-1]$ or feasible to be processed by our proposed algorithm. In either case, the algorithm does not incur significant losses even if a wrong decision is made in the interval $[0,\tl-1]$. Next, consider the case where a job in the optimal solution is processed in the interval $[0, \tl-1]$. If the deadline is no later than $\tl$, the competitive ratio can be safeguarded by the guarantees of the local test, which ensures the performance of our algorithm until time $\tl - 1$. Otherwise, if its deadline is later than $\tl$, it is feasible to schedule it with our algorithm, ensuring the algorithm's performance.


\subsection{Related works}
\noindent\textbf{Online packet scheduling with deadlines.}
The problem of online packet scheduling with deadlines was initially introduced by \citet{hajek2001competitiveness} and \citet{kesselman2004buffer}. There are many studies 
for deterministic algorithms, including both the general model and different special cases \cite{jez2012online, kesselman2004buffer, chin2006online, bohm2019online}. Recently, \citet{vesely2022competitive} introduced an optimal $\phi \approx 1.618$-competitive algorithm for the general model, which matches the lower bound for deterministic algorithms \cite{andelman2003competitive, chin2003online, kesselman2004buffer}.


For the randomized algorithm, different adversary models and results have been explored as presented in studies such as \cite{chin2006online, bienkowski2011randomized, jez2011one}.
Additionally, other variants of the online packet scheduling with deadlines problem have been investigated~\cite{bienkowski2013collecting, bohm2019online}. Notably, in \cite{bohm2019online}, a lower bound of a 1.25-competitive ratio was proven for any constant lookahead, where the algorithm can access future input information for the next few steps.
More information about online packet scheduling with deadlines and related problems can be referred to in survey papers~\cite{goldwasser2010survey, vesely2021packet}.

\noindent\textbf{Learning-augmented algorithms.}
In recent years, a significant body of research has emerged on algorithms with prediction models, which seeks to address the challenge of partial information often available to decision-makers. This line of work utilizes machine-learned predictions about the inputs to design algorithms that achieve stronger bounds (consistency) when provided with accurate predictions. Furthermore, this research aims to maintain a worst-case guarantee (robustness) that holds even in the presence of inaccurate predictions. Optimization problems investigated under this framework include online paging~\cite{DBLP:journals/jacm/LykourisV21,Purohit18improving}, routing~\cite{https://doi.org/10.48550/arxiv.2205.12850}, matching~\cite{Antoniadis20Secretary,https://doi.org/10.48550/arxiv.2206.11397}, and graph problems~\cite{azar2022online,xu2022learning, bernardini2022universal} (see~\cite{MitzenmacherV22} for a survey of works in this area). Many different
scheduling problems have been studied with predictions~\cite{antoniadis2021novel,balkanski2022scheduling,balkanski2023energyefficient,bamas2020learning, Lattanzi2020online,lindermayr2022permutation,Mitzenmacher20Scheduling}. 
See the survey website \cite{algorithms_with_predictions} for more information.
\section{Preliminaries}
\label{sec:prelim}
We first describe the online packet scheduling with deadlines problem. An instance of online packet scheduling with deadlines is described by a collection $\J$ of $n$ unit-length jobs. Each job, denoted as $(r_j, d_j, w_j) \in \J$, is defined by a release time $r_j$, a deadline $d_j$, and a weight $w_j$, where $r_j$ and $d_j$ are integers, and $w_j$ is a non-negative real number. In addition, we assume distinct weights for all packets, as any instance can be adjusted to have unique weights through infinitesimal perturbation while maintaining the competitive ratio. Let $T$ be a finite time horizon of the problem. At each time $t \in [T]$, only one job can be processed, and it must satisfy $r_j \le t$ and $d_j-1 \ge t$; that is, a job can only be processed after its arrival and before its deadline. 
The objective is to process a set of jobs of maximum
total weight. We also let $\J_{\le t}$ be a set of jobs in $\J$ with release time at most $t$. 

We refer to a collection of processed jobs in the interval $[0, T]$ as a \emph{schedule}.
Additionally, we let $\S = (\S_1, \cdots, \S_T)$ denote a set of \emph{choices} for processing jobs.
To determine $\mathcal{S}_i$ at each time $i \in [T]$, we consider the information of all available jobs, including their arrival times, deadlines, and weights.
For any realization $\J$ of the job instance, we define $\S(\J)$ as the corresponding jobs processed according to $\S$ (adding zero-weight dummy jobs if there are no jobs to be processed). Thus, $\S(\J)$ represents a schedule for the job instance $\J$.
Finally, we let $W(\S(\J)) = \sum_{j\in \S(\J)} w_j$ be the total weight of the jobs in $\S(\J)$.
\begin{definition}[Dominance Relation \cite{jez2012online}]
For any two jobs $j$ and $j'$ which are released but not yet processed at a certain time, $j$ is said to dominate $j'$ if 
$w_{j} > w_{j'}$ and $d_{j} \le d_{j'}$.

\end{definition}
According to \cite{jez2012online}, any schedule can be arranged in canonical order in the following way. Assign the earliest deadline feasible (i.e., a job that has been released but not yet processed and is not dominated by any other job) to each successive step. The resulting schedule is referred to as a \emph{canonical schedule}. For any job instance $\J$, let $\OPT(\J)$ be the optimal schedule for $\J$ and $\OPT(\J_{\le t})$ be the optimal schedule for $\J_{\le t}$, respectively. 
In the following, we let an optimal schedule be a canonical schedule. 


We also let $\OPT(\J)^{(t)}$  and $\OPT(\J)^{(\le t)}$ be the job processed at time $t$ and the set of jobs processed in the interval $[0,t]$ of $\OPT(\J)$, respectively (add a dummy zero-weight job if there is no job processed).
Likewise, for any schedule $\S(\J)$, we let $\S(\J)^{(t)}$  and $\S(\J)^{(\le t)}$ denote the job processed 
at time $t$ and the set of jobs processed in the interval $[0,t]$ of $\S(\J)$, respectively.
Also, for any algorithm $\A$ and job instance $\J$, we let $\A(\J)$ be the output schedule of running $\A$ on $\J$. 
Finally, for any job instance $\J$, we let $\D (t, \J)$ be the set of expired jobs at time $t$; that is, $d_j < t + 1$ and which have not been processed. 

\noindent\textbf{The learning-augmented online packet scheduling with deadlines problem.} The first important modeling task in the learning-augmented framework is to choose the prediction. To answer this question, we augment the online packet scheduling with deadlines problem with predictions regarding future job arrivals, and we call this problem the learning-augmented online packet scheduling with deadlines problem. The algorithm is given a prediction $\predJ = \{(\hat{r_j},\hat{d_j}, \hat{w_j})\}$ regarding the jobs $\J = \{(r_j, d_j, w_j)\}$ that arrive online (note that we do not necessarily have $|\J| = |\predJ|$). Here, one may consider a prediction scenario in which, at time $t$, the algorithm has access to extra information about arrivals in a few future steps, essentially representing a constant lookahead setting. However, according to \cite{bohm2019online}, there exists a strong 1.25-competitive lower bound for any constant lookahead. Therefore, to achieve better consistency, we adopt a model in which the algorithm is given with predictions about the entire input sequence in the beginning.

Another important modeling task in the learning-augmented framework is to design an appropriate measure for the prediction error. With the definition of \emph{schedule} and a set of \emph{choices}, we denote $\hat{\S}$ as the optimal choices for the predicted set of jobs $\predJ$, represented as $\hat{\S}(\predJ) = \OPT(\predJ)$.
We will now define the prediction error,
$$\eta(\J, \predJ): = \max_{t\in [T]} \frac{W(\OPT(\J_{\le t})^{(\le t)})}{W(\hat{\S}(\J)^{(\le t)})} \ ,$$ 
as the maximum ratio between the total weight of the processed jobs in the interval $[0, t]$ of the optimal schedule for all released jobs until time $t$ and the processed weights obtained by following the choices of the optimal schedule for the prediction until time $t$, considering all possible values of $t$. In the following, we use $\eta$ to denote the error when $\J$ and $\predJ$ are clear. 

\noindent\textbf{Choice of the prediction error.}
We define the prediction error as the $\ell_\infty$ norm of the ratio between the total weight of the processed jobs in the interval $[0,t]$ of the optimal schedule for all released jobs until time $t$ and the processed weights by following the choices of the optimal schedule for the prediction until time $t$, considering all possible $t$. Other types of prediction errors, such as the $\ell_1$ norm of the weight difference between the prediction and the realization, may also be considered, but there are tradeoffs between different errors. For instance, the $\ell_1$ norm does not distinguish between instances with one large prediction error and those with many small prediction errors. The $\ell_\infty$ norm, on the other hand, does not distinguish between one large prediction error and many large prediction errors.

The reason for choosing this error measurement is that the competitive ratio is highly sensitive to both the weight and deadline of jobs, making it crucial to capture their effects accurately. Therefore, we use the $\ell_\infty$ norm of all time steps to measure the prediction error. We believe that our choice of error captures the essence of the problem. However, we also recognize that other error measurements could be explored as interesting directions for future research. 
For instance, our error measure is sensitive to small shifts in the job's release time, such as when the realization occurs one time unit earlier than predicted. In this case, although the predictions capture the pattern of realization, it results in a large error based on the proposed error measure.
It would be worthwhile to investigate an alternative error measure as in \cite{bamas2020learning} that can accommodate small shifts in job release time.

\noindent\textbf{Performance metrics.} The standard evaluation metrics for an online algorithm with predictions are its consistency, robustness, and competitive ratio as a function of the prediction error \cite{DBLP:journals/jacm/LykourisV21}.
We let a function $c$ of error
$\eta$, where $\eta \ge 1$ using a given predictor such that $c(\eta)$ represents the worst-case competitive ratio given that the error is at most $\eta$.
We say an algorithm $\A$ is $\alpha$-consistent if $c(1) \le \alpha$ (competitive ratio when the prediction is exactly correct) and we say $\A$ is $\beta$-robust if for all $\eta \ge 1, c(\eta) \le \beta$ (competitive ratio when the error is arbitrarily large). We say that the competitive ratio of $\A$ is smooth if it smoothly degrades from $\alpha$ to $\beta$ as the prediction error $\eta$ grows.



\section{Algorithms}
\label{sec:algo_deadlines}

\paragraph{Warm up: A simple consistent, smooth, non-robust algorithm.}
We first show that an algorithm that blindly follows the 
predictions is $\eta$-competitive. Clearly, when $\J = \predJ$, since the algorithm follows the choices for the prediction, it returns an optimal schedule. However, this approach is not robust, as the competitive ratio can be arbitrarily large in case of incorrect predictions.
\begin{algorithm}[H]
\caption{\textsc{A simple 1-consistent, smooth algorithm but not robust} }
\label{alg:pred}
\begin{algorithmic}[1]
\INP A set of predicted jobs $\predJ$.
\STATE Compute a set of optimal choices $\hat{\S}$ for the predicted jobs $\predJ$, where $\hat{\S}(\predJ) = \OPT(\predJ)$
\STATE \textbf{return} $\hat{\S}(\J)$
\end{algorithmic}   
\end{algorithm}

\begin{lemma}
    Algorithm~\ref{alg:pred} is $\eta$-competitive.
\end{lemma}

\begin{proof}
Since $\hat{\S}(\J)$ follows the same choices as $\OPT(\predJ)$, we have 
\begin{align*}
    \frac{W(\OPT(\J))}{W(\hat{\S}(\J))} &\le \max_{t\in [T]} \frac{W(\OPT(\J_{\le t})^{(\le t)})}{W(\hat{\S}(\J)^{(\le t)})}  = \eta\ ,
\end{align*}
where the last equality comes from the definition of $\eta$.
\end{proof}
The main problem with Algorithm~\ref{alg:pred} is its lack of robustness. Particularly, the competitive ratio can be unbounded if the prediction is completely wrong. Our next step is to robustify the algorithm to achieve a bounded competitive ratio and maintain its consistency at the same time.

\begin{figure*}
    \centering 
    \includegraphics[width = 10 cm]{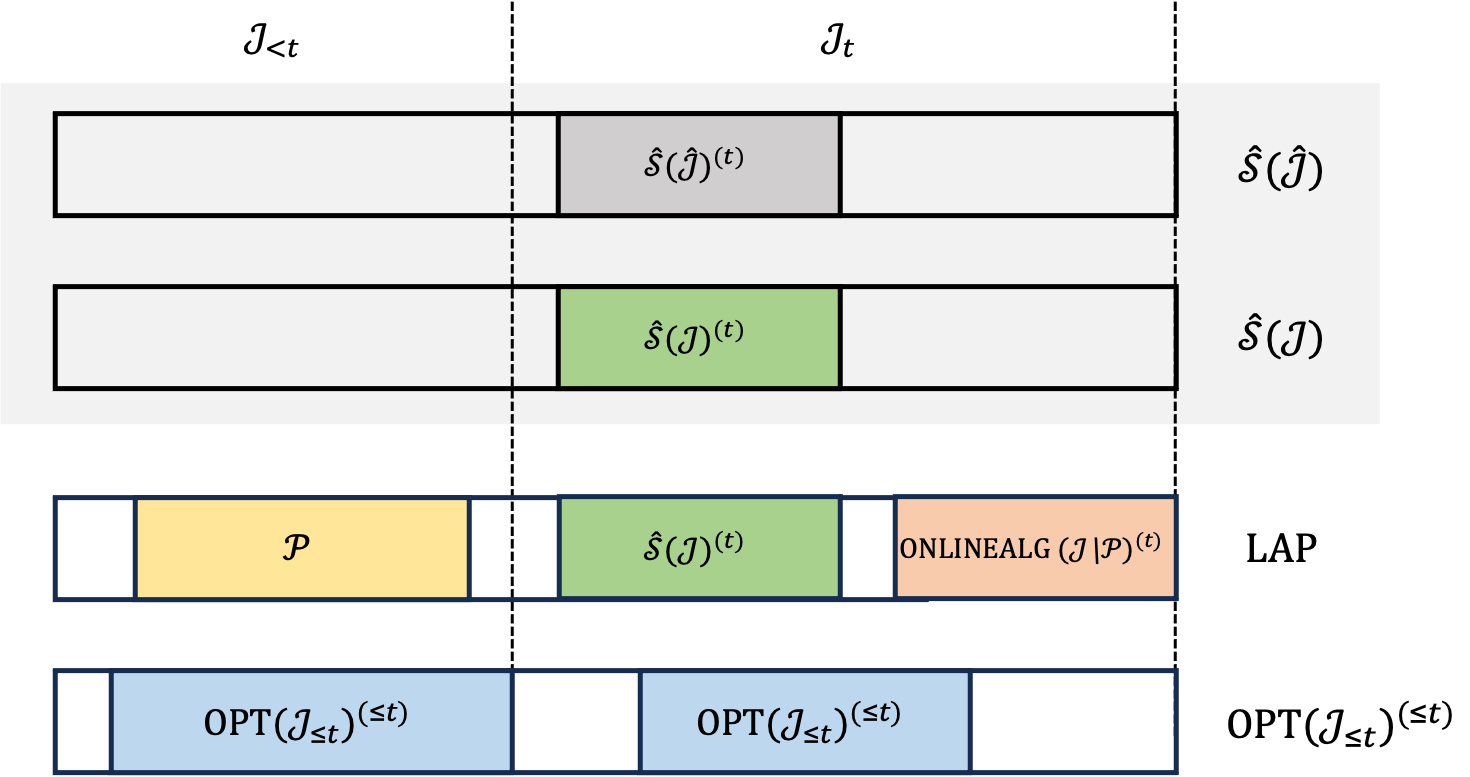} 
    \caption{Illustration of one iteration of the local test (Lines 4-7) of the \laps\ algorithm. The first two rows with a gray background represent the job following the predicted schedule $\hat{\S}(\predJ)^{(t)}$ and the corresponding job in realization $\hat{\S}(\J)^{(t)}$ at time $t$, while the last two rows indicate the processed jobs $\P$ (yellow) as well as two potential operations for \laps\ and the jobs in $\OPT(\J_{\le t})^{(\le t)}$. 
    For each local test at time $t$, \laps \ compares $W(\OPT(\J_{\le t})^{(\le t)})$ with $W(\P \cup \hat{\S}(\J)^{(t)})$ to decide whether \laps \ processes either $\hat{\S}(\J)^{(t)}$ (green) or $\textsc{OnlineAlg}(\J \setminus \P )^{(t)}$  (orange). 
    In this example, we assume $\hat{\S}(\J)^{(t)} \notin \P$ and $\D (t,\J) = \emptyset$.}
    \label{fig:laps}
 \end{figure*}

\begin{algorithm}[H]
    \caption{\textsc{Learning-Augmented Online Packet Scheduling with Deadlines} \ (\laps)}
    \label{alg-general}
\begin{algorithmic}[1]
 \INP A set of predicted jobs $\predJ$, an online algorithm \textsc{OnlineAlg}, and a threshold $\rho \ge 1$.
    \STATE Compute a set of optimal choices $\hat{\S}$ for the predicted jobs $\predJ$, where $\hat{\S}(\predJ) = \OPT(\predJ)$
    \STATE $\P = \emptyset$
    \STATE \textbf{for} $t \ge 0$ \textbf{do}
    \STATE \quad \textbf{if} $\hat{\S}(\J)^{(t)} \notin \P$ \textbf{and} $\frac{W(\OPT(\J_{\le t})^{(\le t)})}{W(\P \cup \hat{\S}(\J)^{(t)})} \le \rho$ \textbf{then} 
    \STATE \quad \quad $\P = \P \cup \hat{\S}(\J)^{(t)}$
    \STATE \quad \textbf{else}
    \STATE \quad \quad $\P= \P \cup \textsc{OnlineAlg}(\J \setminus (\P \cup \D (t,\J)) )^{(t)}$
    \STATE \textbf{return} $\P$
\end{algorithmic}
\end{algorithm}

\subsection{Description of the algorithm}

The proposed algorithm, known as \laps, is a general algorithm designed for the online packet scheduling problem with deadlines. The algorithm takes a set of predicted jobs $\predJ$, a threshold $\rho$, and an online algorithm \textsc{OnlineAlg} that solves the online packet scheduling with deadlines problem without predictions as inputs. Initially, \laps \ computes a set of choices of the optimal schedule for the predicted jobs $\predJ$, namely, $\hat{\S}$ (Line 1). 
At each time $t$, \laps \ performs a local test to determine whether to follow $\hat{\S}$ or switch to \textsc{OnlineAlg} (Lines 4-7). 
The local test first checks whether $\hat{\S}(\J)^{(t)} \in \P $. If not, \laps \ computes the ratio of the total weight of the processed jobs in the interval $[0,t]$ for the optimal schedule of all the released jobs until time $t$ and the current processed weights 
until time $t-1$, with the addition of the weight of the processed job corresponding to the predicted choice at time $t$.
See \cref{fig:laps} for an illustration of an iteration of the local test.
Note that \laps \ could switch back and forth between following the prediction or \textsc{OnlineAlg}.

The underlying principle of \laps \ is to switch to \textsc{OnlineAlg} when $\hat{\S}(\J)^{(t)} \in \P $ or the current ratio is larger than the threshold ($> \rho$) to prevent the competitive ratio from worsening.
However, if $\hat{\S}(\J)^{(t)} \notin \P $ and the current ratio is no greater than $\rho$, the algorithm will continue to trust the prediction and follow $\hat{\S}$ until a significant deviation occurs. The formal description of \laps \ is provided in Algorithm~\ref{alg-general}.

\subsection{Analysis of the algorithm}
The following paragraph analyzes the performance of \laps. Initially, we demonstrate that \laps \ is 1-consistent (Lemma~\ref{lem:cons}). Subsequently, we establish that \laps \ is $\eta$-competitive under the condition that $\eta \le \rho$ (Lemma~\ref{lem:smooth}). Finally, we investigate the scenario where $\eta > \rho$ and demonstrate that \laps \ maintains a bounded competitive ratio of $ \rho + \gamma_{\texttt{on}} + 1$ (Lemma~\ref{lem:robust}).

We begin with an important observation of the optimal solution.
\begin{lemma}\label{lem:localtest}
For every $t\in [T]$, $W(\OPT(\J)^{(\le t)}) \ge W(\OPT(\J_{\le t})^{(\le t)})$.  
\end{lemma}
\begin{proof}
The core concept of the proof is that the optimal schedule prioritizes processing jobs with later deadlines only if there are other jobs with significant weights arriving in the future with the same deadline. Assume that $W(\OPT(\J)^{(\le t)}) < W(\OPT(\J_{\le t})^{(\le t)})$. This implies the existence of a pair of distinct jobs $j$ and $j'$, where $w_j > w_{j'}$, such that $j \in \OPT(\J_{\le t})^{(\le t)}$ and $j' \in \OPT(\J)^{(\le t)}$.

First, consider the case where $j \notin \OPT(\J)^{(>t)}$. Replacing $j'$ with $j$ would increase $W(\OPT(\J)^{(\le t)})$ without affecting $W(\OPT(\J)^{(>t)})$, contradicting the optimality of $\OPT(\J)$.

Next, consider the case where $j \in \OPT(\J)^{(>t)}$. In this case, there must exist a job $j'' \in \OPT(\J_{\le t})^{(>t)}$ such that $j'' \notin \OPT(\J)^{(>t)}$. Otherwise, scheduling $j$ in $\OPT(\J_{\le t})^{(>t)}$ and scheduling $j'$ in $\OPT(\J_{\le t})^{(\le t)}$ would increase $W(\OPT(\J_{\le t})^{(\le t)})$, contradicting the optimality of $\OPT(\J_{\le t})$.

Consider the relationship between $j'$ and $j''$. According to our assumption, each job has a unique weight, where $w_{j'} \neq w_{j''}$. In this case, replacing $j$ and $j''$ with $j'$ and $j$ (or vice versa) would increase the total weight of either $\OPT(\J)$ or $\OPT(\J_{\le t})$, contradicting the optimality of either $\OPT(\J)$ or $\OPT(\J_{\le t})$, and we thus complete the proof.

\end{proof}

\begin{lemma}\label{lem:cons}
\laps \ is 1-consistent.
\end{lemma}
\begin{proof}
To prove 1-consistency, one needs to ensure that \laps\ satisfies the condition in Line 4 for all $t\in [T]$. To achieve this, we leverage the insights provided by~\cref{lem:localtest}.
When $\J = \predJ$, we have $\hat{\S}(\predJ) = \OPT(\J)$, and the condition of Line 4 can be represented as: 
$$ \frac{W(\OPT(\J_{\le t})^{(\le t)})}{W(\OPT(\J)^{(\le t)})} \le \rho \ .$$ Since $\rho \ge 1$, we have $$\frac{W(\OPT(\J_{\le t})^{(\le t)})}{W(\OPT(\J)^{(\le t)})} \le 1 \le \rho \ , $$ which satisfies the condition of Line 4. Therefore, the algorithm will follow $\hat{\S}$ from the beginning to the end, thus completing the proof.
\end{proof}





\begin{lemma}\label{lem:smooth}
    \laps \ is $\eta$-competitive if $\eta \le \rho$. 
\end{lemma}
    
\begin{proof}
Initially, by applying the definition of $\eta$, we have 
\begin{align}\label{equ:1}
    \frac{W(\OPT(\J))}{W(\hat{\S}(\J)) }
&\le \max_{t\in [T]} \frac{W(\OPT(\J_{\le t})^{(\le t)})}{W(\hat{\S}(\J)^{(\le t)})}  = \eta \ .
\end{align}

Next, by our assumption $\eta \le \rho$, we are in the case where \laps \ follows $\hat{\S}$ from the beginning to the end (Lines 1-5). Therefore, we have
\begin{align*}
\frac{W(\OPT(\J))}{W(\P)} & = \frac{W(\OPT(\J))}{W(\hat{\S}(\J))}  \le \eta\ ,
\end{align*}
where the second inequality comes from the result of (\ref{equ:1}).
\end{proof}


To prove the robustness of \laps, we partition the output schedule $\P$ into $\P^{(\le \tl -1)}$ and $\P^{(\ge \tl)}$ as processed jobs of \laps \ within the intervals $[0,\tl-1]$ and $[\tl, T]$, respectively, where $\tl - 1 = \argmax_{t: t \in [T]} \frac{W(\OPT(\J_{\le t})^{(\le t)})}{W(\P^{(\le t )})} \le \rho $ is the last time that the local test ratio is not greater than the threshold $\rho$. In this case, \laps \ follows \textsc{OnlineAlg} from $\tl$ to the end.





\begin{lemma}\label{lem:maximize}
$\OPT(\J)^{(\le \tl-1)} \cap \OPT(\J_{\le \tl-1})^{(> \tl-1)} \neq \emptyset$, if $W(\OPT(\J)^{(\le \tl-1)}) > W(\OPT(\J_{\le \tl-1})^{(\le \tl-1)})$.
\end{lemma}


\begin{proof}
We prove the lemma by contradiction. Assume $\OPT(\J)^{(\le \tl-1)} \cap \OPT(\J_{\le \tl-1})^{(> \tl-1)} = \emptyset$. By replacing $\OPT(\J_{\le \tl-1})^{(\le \tl-1)}$ with $\OPT(\J)^{(\le \tl-1)}$, we have 
\begin{align*}
    &W(\OPT(\J)^{(\le \tl-1)}) + W(\OPT(\J_{\le \tl-1})^{(> \tl-1)}) \\
    &> W(\OPT(\J_{\le \tl-1})^{(\le \tl-1)}) + W(\OPT(\J_{\le \tl-1})^{(> \tl-1)}) \\
    &= W(\OPT(\J_{\le \tl-1}))\ ,
\end{align*}
which contradicts the optimality of $\OPT(\J_{\le \tl-1})$.
\end{proof}


\begin{restatable}{theorem}{lemrobust}
\label{lem:robust}
\laps \ is $\rho + \gamma_{\texttt{on}} + 1$-robust.
\end{restatable}
\begin{proof}[Proof sketch; full proof in Appendix]
We approach the proof by considering two different cases of the optimal solution, each addressing different scenarios of the relationship between the processed weights after time $\tl - 1$ and the processed weights on or before time $\tl-1$, where $c > 0$ is a parameter that we choose later:

\textbf{Case 1:} $W(\OPT(\J)^{(> \tl-1)}) > c W(\OPT(\J)^{(\le \tl-1)})$
    
    In this case, the processed weights after time $\tl-1$ significantly exceed the processed weights on or before time $\tl-1$. Since \laps \ follows \textsc{OnlineAlg} after time $\tl$, we can utilize the competitiveness of \textsc{OnlineAlg} to show that the competitive ratio of this case is $\frac{c+1}{c}\gamma_{\text{on}}$.

\textbf{Case 2:} $W(\OPT(\J)^{(> \tl-1)}) \le c W(\OPT(\J)^{(\le \tl-1)})$

    We analyze this case by looking at the following two sub-cases:

    \begin{enumerate}
        \item[(a)] $W(\OPT(\J)^{(\le \tl-1)}) = W(\OPT(\J_{\le \tl-1})^{(\le \tl-1)}).$

        For this case, since the processed weights of \laps \ in the interval of $[0, \tl-1]$ is at least $\frac{1}{\rho} W(\OPT(\J_{\le \tl-1})^{(\le \tl-1)})$, we can obtain the competitive ratio of this case as $(c+1)\rho$. 

        \item[(b)] $W(\OPT(\J)^{(\le \tl-1)}) > W(\OPT(\J_{\le \tl-1})^{(\le \tl-1)})$

        In this sub-case, there is an overlapping scenario. By utilizing Lemma~\ref{lem:maximize}, we can handle the overlapping part. Finally, we can show that the competitive ratio is $(c+1)(\rho + 1)$.
    \end{enumerate}

By optimizing $c = \frac{\gamma_{\text{on}}}{\rho+1}$, we conclude that the overall competitive ratio is $\gamma_{\text{on}} + \rho + 1$.
\end{proof}

We are ready to state the main result of this section, which is our upper bound on the competitive ratio of \laps.

\begin{theorem}\label{thm:mainthm} 
\laps{} with a $\gamma_{\texttt{on}}$-competitive
algorithm \textsc{OnlineAlg} and an optimal solution for the prediction $\predJ$ achieves a competitive ratio of
\[
\begin{cases*}
    \eta & \text{if $\eta \le \rho$} \\
    \rho + \gamma_{\texttt{on}} + 1 & \text{if $\eta > \rho$} \ ,
\end{cases*}
\]
where 
$\rho \ge 1$ is an error tolerance parameter chosen by \laps.
\end{theorem}
\section{Lower bound}
\label{sec:lower-bound}
This section shows some impossibility results about learning augmented algorithms in the context of \emph{Online Packet Scheduling with Deadlines}.
The proofs in this section can be found in \cref{app:lb}.

We first show that there exists an instance for which any $1$-consistent algorithm cannot be $2$-robust. 
\begin{restatable}{rPro}{proplb}
\label{prop:1-consist}
Assume a deterministic learning-augmented algorithm is 1-consistent; then, the competitive ratio cannot be better than 2.
\end{restatable}






Subsequently, we show that there is a necessary trade-off between consistency and robustness for the deadline model by extending the previous instance.

\begin{restatable}{theorem}{thmlb}
\label{LB-deadlines}
Assume a deterministic learning-augmented algorithm is $1+\alpha$-consistent; then, the competitive ratio cannot be better than $\frac{2}{1+\alpha}$.    
\end{restatable}

\section{Experiment results}
\label{sec:exp}
We empirically evaluate the performance of the proposed algorithm \laps \ on both synthetic and real datasets against benchmarks without using predictions. 
We specifically consider the agreeable deadline model, where jobs with later arrival times cannot have earlier deadlines. 
%
The choice is motivated by an experimental study in a related paper \cite{sakr2016empirical}. Additionally, the model provides the same theoretical guarantees as the general deadline model. Furthermore, the current best algorithm for the general setting is complex, making the agreeable deadline model a pragmatic and efficient alternative for our analysis.

\subsection{Experiment settings}
\noindent\textbf{Benchmarks.} \laps{} is Algorithm~\ref{alg-general} with the default setting $\rho = 0.1$. \textbf{\textsc{MG}} is the Modified Greedy Algorithm from \cite{jez2012online}, which achieves a tight 1.618-competitive ratio for the agreeable deadline setting. \textbf{\textsc{Greedy}} is the Greedy Algorithm from \cite{kesselman2004buffer}, which always processes a packet with the largest weight in the buffer in each step and has been proven to be 2-competitive. \textbf{\textsc{EDF}} is the heuristic earliest-deadline-first algorithm, where the algorithm processes the packet with the earliest deadline in the buffer of each step. Despite the absence of theoretical guarantees on this problem, it performs well in many scheduling scenarios. Finally, we consider \textbf{\textsc{EDF}$_{\mathrm{\alpha}}$} from \cite{chin2006online}. For each time step, the algorithm first identifies the maximum weight $h$ of all packets in the buffer and processes a packet with a weight at least $\alpha \cdot h$ and the earliest deadline. It has been proven that the competitive ratio is no better than 2 in the agreeable deadline setting \cite{jez2012online}.


\noindent\textbf{Data sets.}
We consider two synthetic datasets and two real datasets. For the synthetic data, we begin by generating a set of real jobs denoted as $\J$. To create the predicted set of jobs $\predJ$, we introduce some error $err(j)$ for either weight or deadline, examining the effects on both weight and deadline. The predicted set of jobs is then defined as: (1) $\predJ= \{(j,r_j, d_j, w_j + err(j): j \in \predJ\}$, where $err(j)$ is sampled i.i.d. from $\N(0, \sigma)$ and (2) $\predJ= \{(j,r_j, d_j  + err(j), w_j: j \in \predJ\}$ where $err(j)$ is uniformly random sampled i.i.d. from the range $[-k,k]$, and $k$ is an integer.

For the first synthetic dataset, we generate predictions using a uniform distribution. Specifically, for each time $t\in {1,\ldots, T}$, where we set $T=75$, the number of job arrivals at time $t$ follows a uniform distribution with a range of [2,8]. For the second synthetic dataset, predictions are generated using a power-law distribution. More precisely, for each time $t\in {1,\ldots, T}$ (with $T=75$ fixed), the number of job arrivals at time $t$ is determined by $M(1-p(a))$, where $p(a)$ is sampled from a power-law distribution with parameter $a$, and $M$ is a scaling parameter. Throughout all experiments, we consistently use the values $a = 150$ and $M = 500$.
 
We also evaluate all the benchmarks on the College Message dataset \cite{panzarasa2009patterns} and the Brightkite dataset \cite{cho2011friendship} from the SNAP database. 
The college dataset spans 23 days, each with between 300 and 500 data points, while the Brightkite dataset covers 10 days within this range.
%
We note that the first dataset was previously used in the context of learning augmented
algorithms by \cite{balkanski2023energyefficient} and the second one was used in the \cite{DBLP:journals/jacm/LykourisV21} and \cite{bamas2020learning}.

We first fix the error parameter $\sigma>0$, then, for each day, we define the true set $\J$ as the arrivals for this day, and we create the predictions $\predJ$ by adding some error $err(i)$ to either weights or deadlines of each job $i$, where $err(i)$ is sampled i.i.d. from $\N(0, \sigma$).

\vspace{-2.5mm}
\subsection{Experiment results}

\begin{figure*}[t!]
 	\centering
  \setlength{\belowcaptionskip}{-8pt}
    \subfigure{\includegraphics[width=0.235\textwidth]{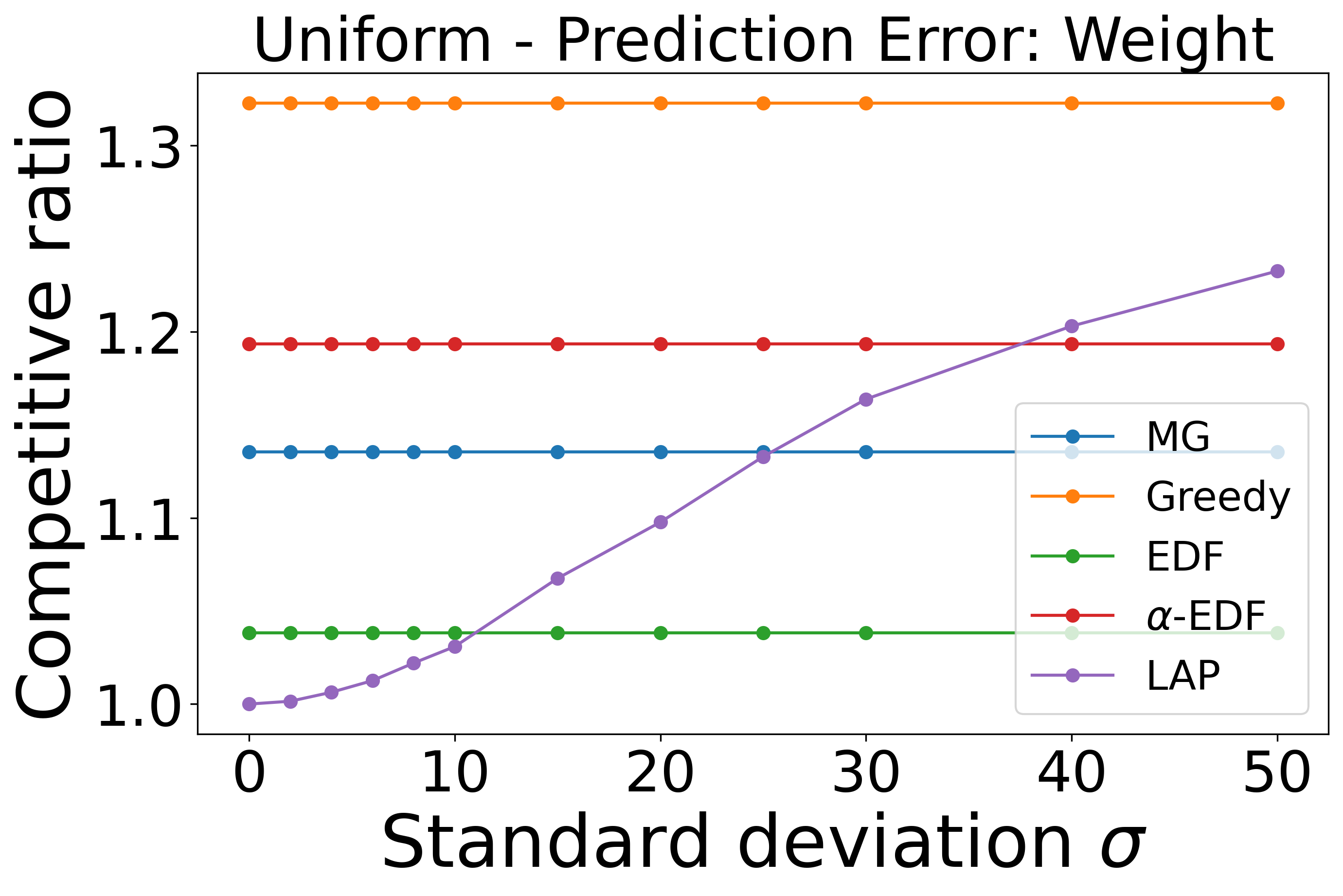}} 
    \subfigure{\includegraphics[width=0.24\textwidth]{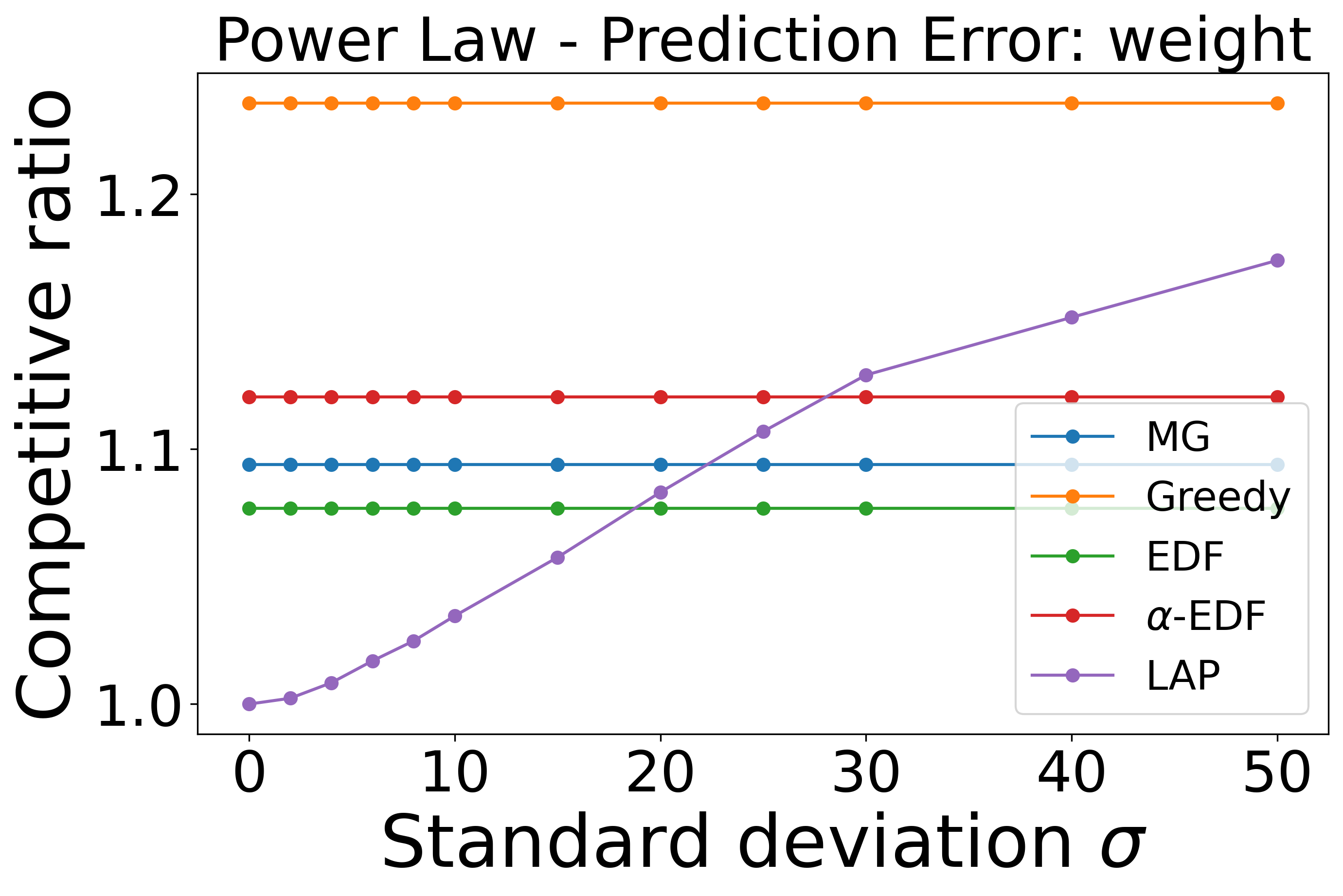}} 
    \subfigure{\includegraphics[width=0.24\textwidth]{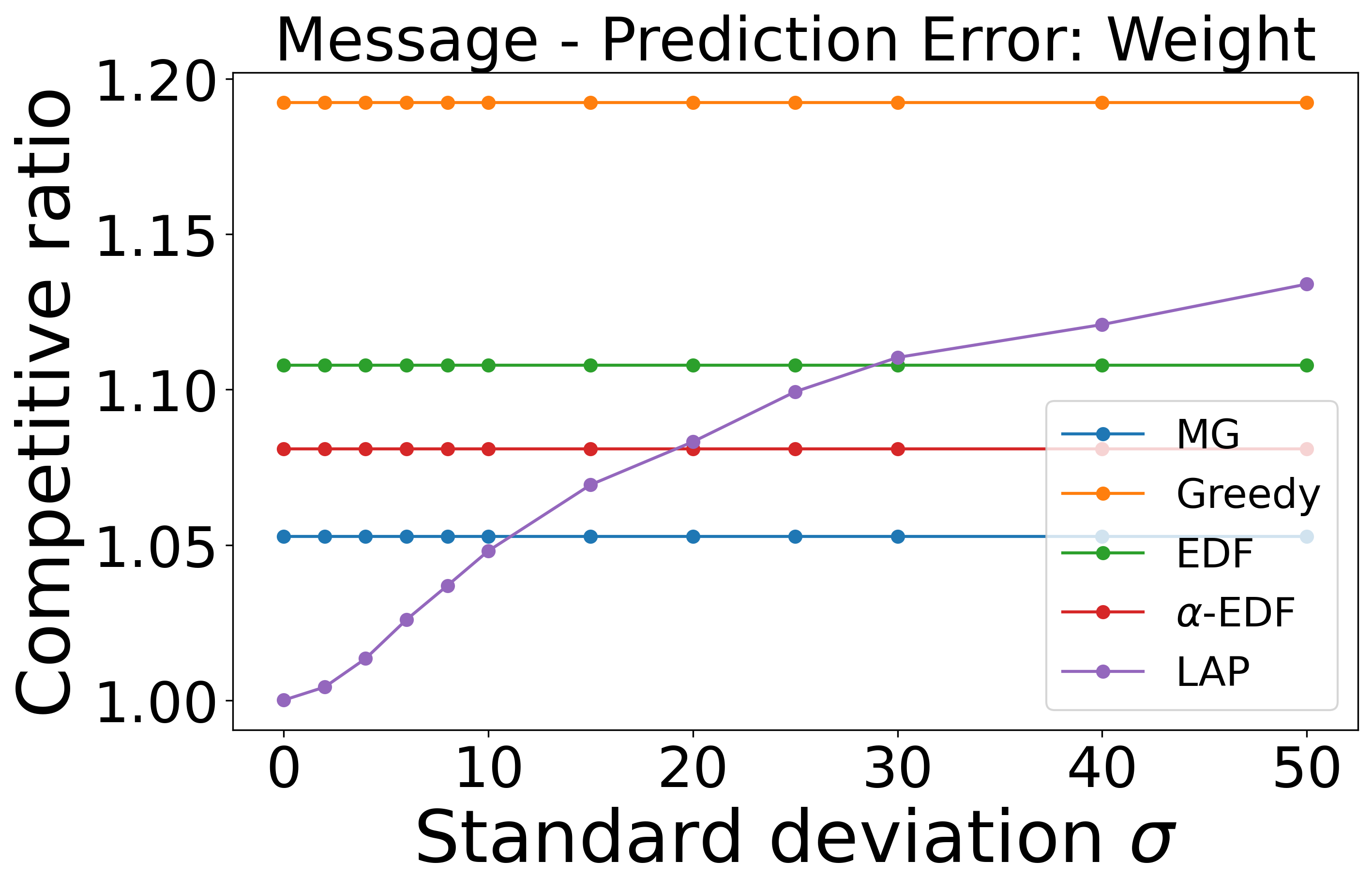}} 
    \subfigure{\includegraphics[width=0.24\textwidth]{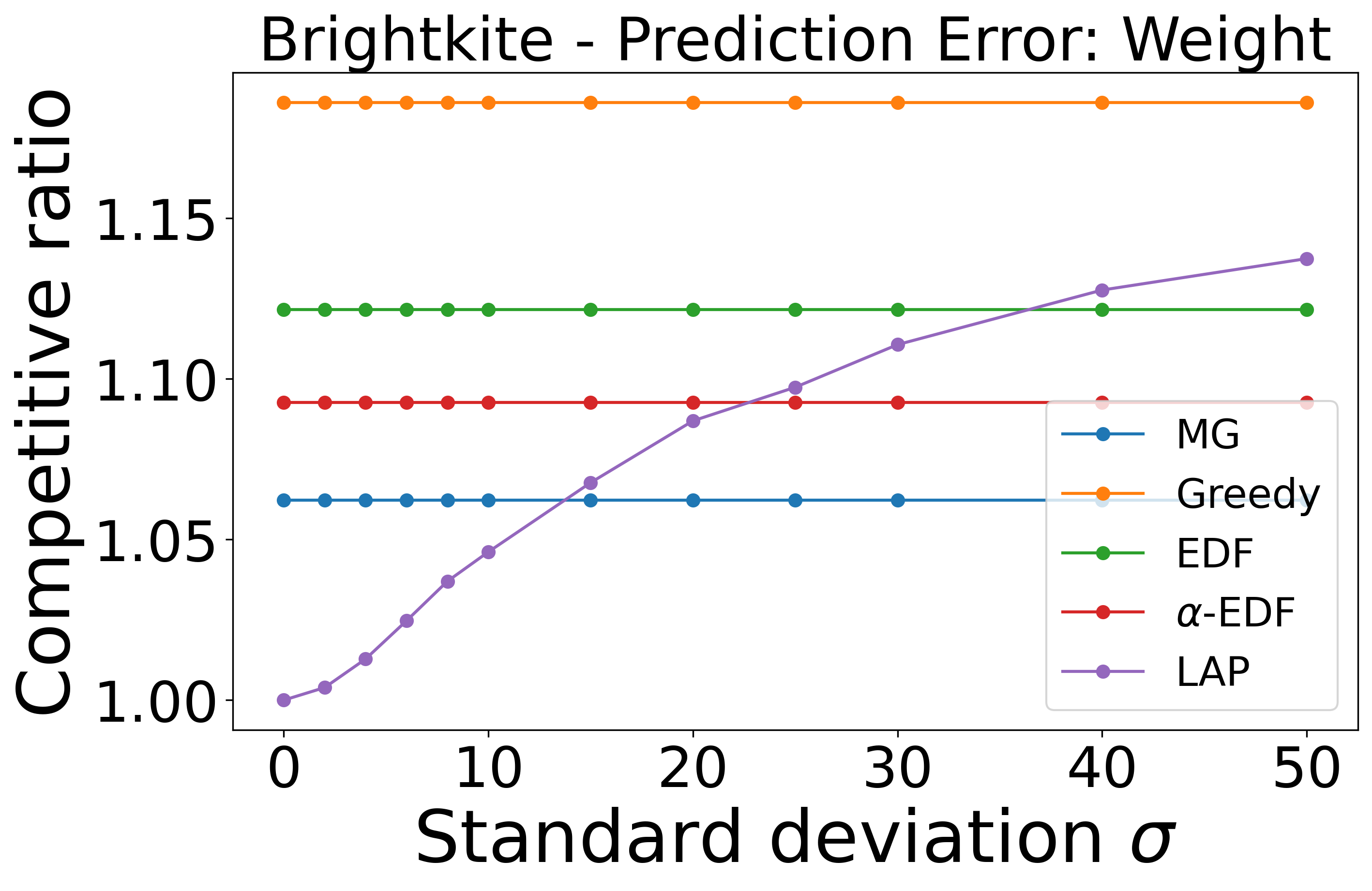}}
    \caption{\footnotesize{The competitive ratio achieved by our algorithm, \laps, and the benchmark algorithms, as a function of the error parameter $\sigma$.}}
    \label{fig:exp:1}
 \end{figure*}

\begin{figure*}[t!]
 	\centering
  \setlength{\belowcaptionskip}{-8pt}
    \subfigure{\includegraphics[width=0.235\textwidth]{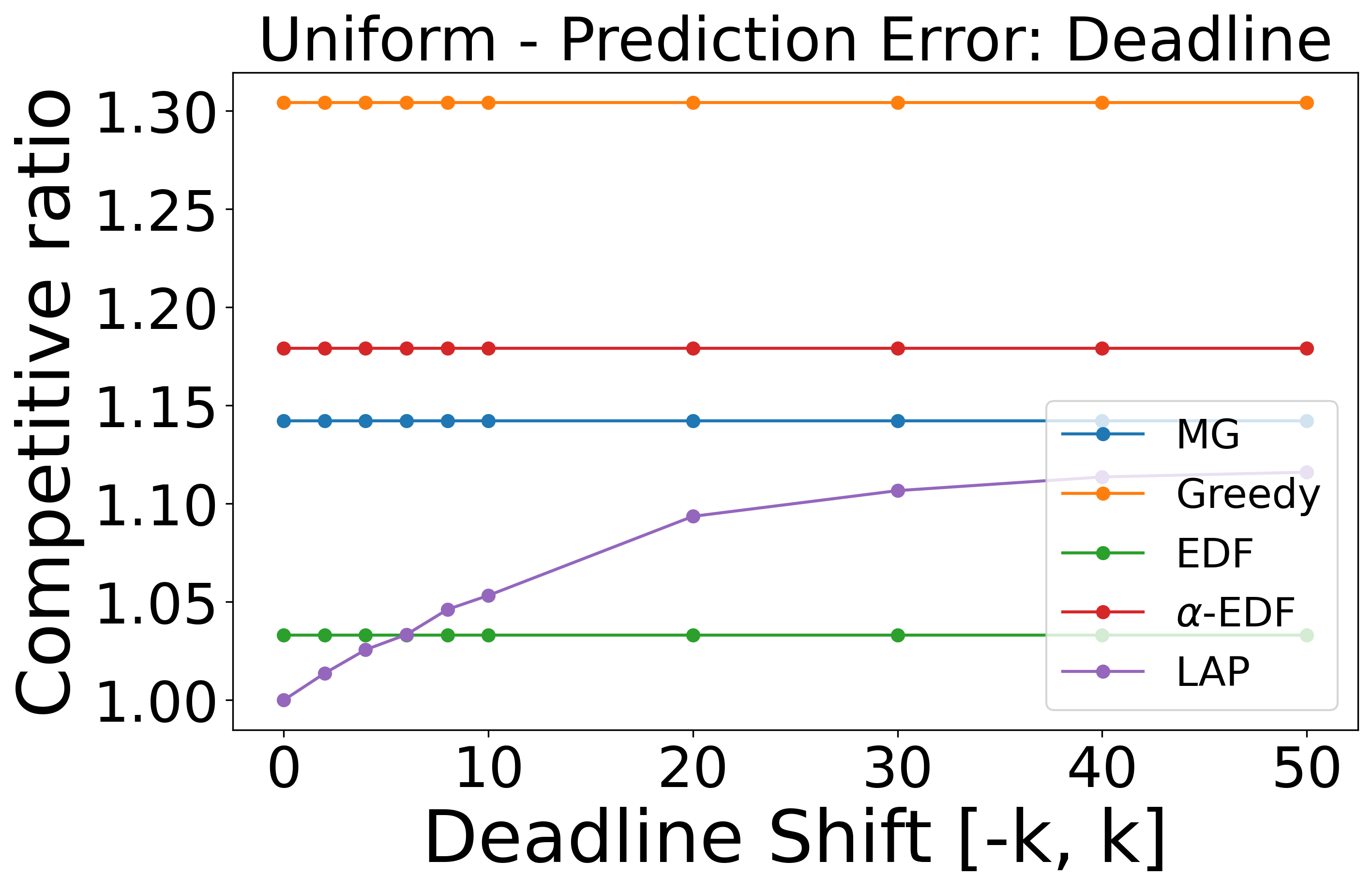}}
    \subfigure{\includegraphics[width=0.24\textwidth]{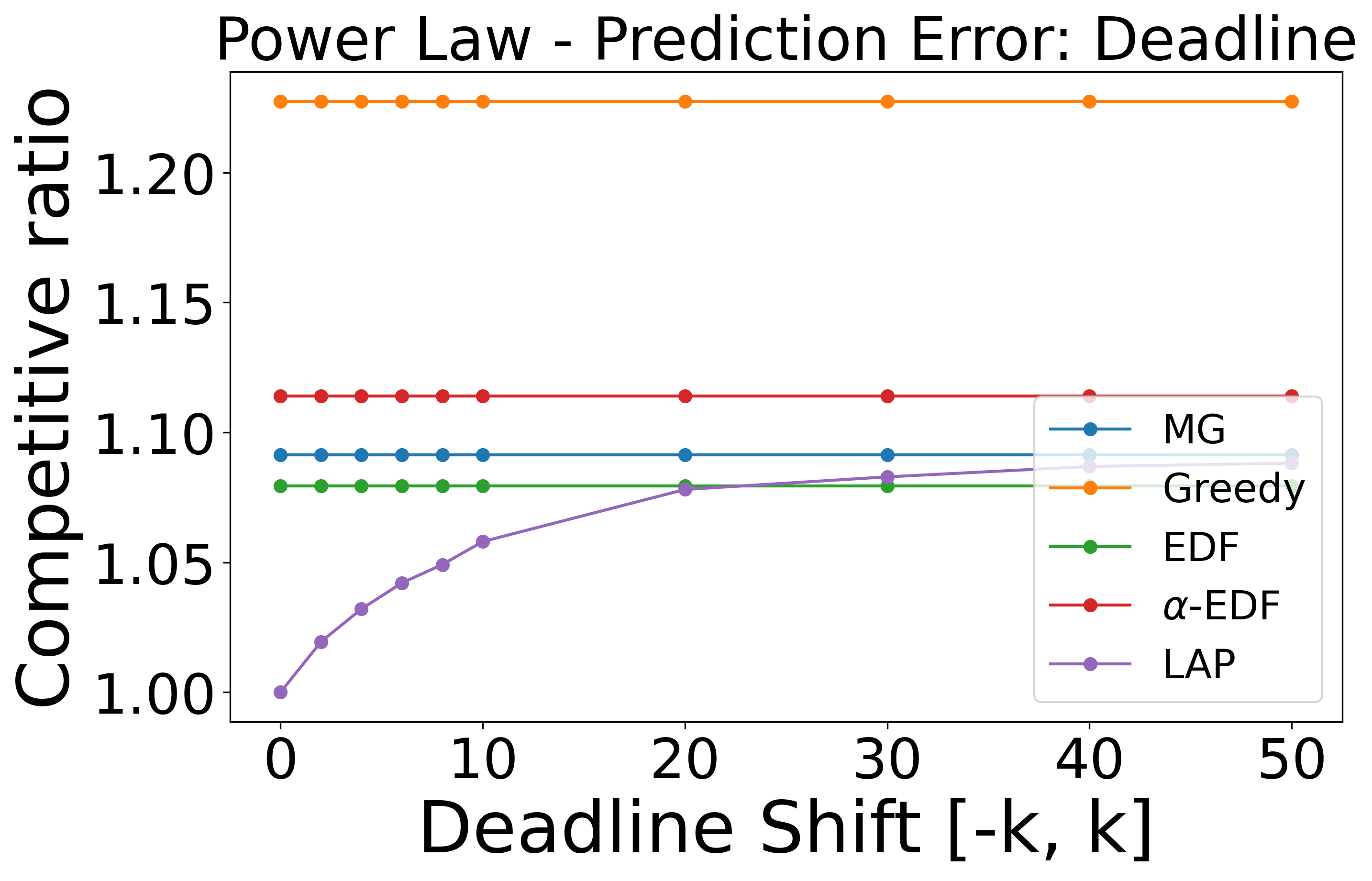}} 
    \subfigure{\includegraphics[width=0.24\textwidth]{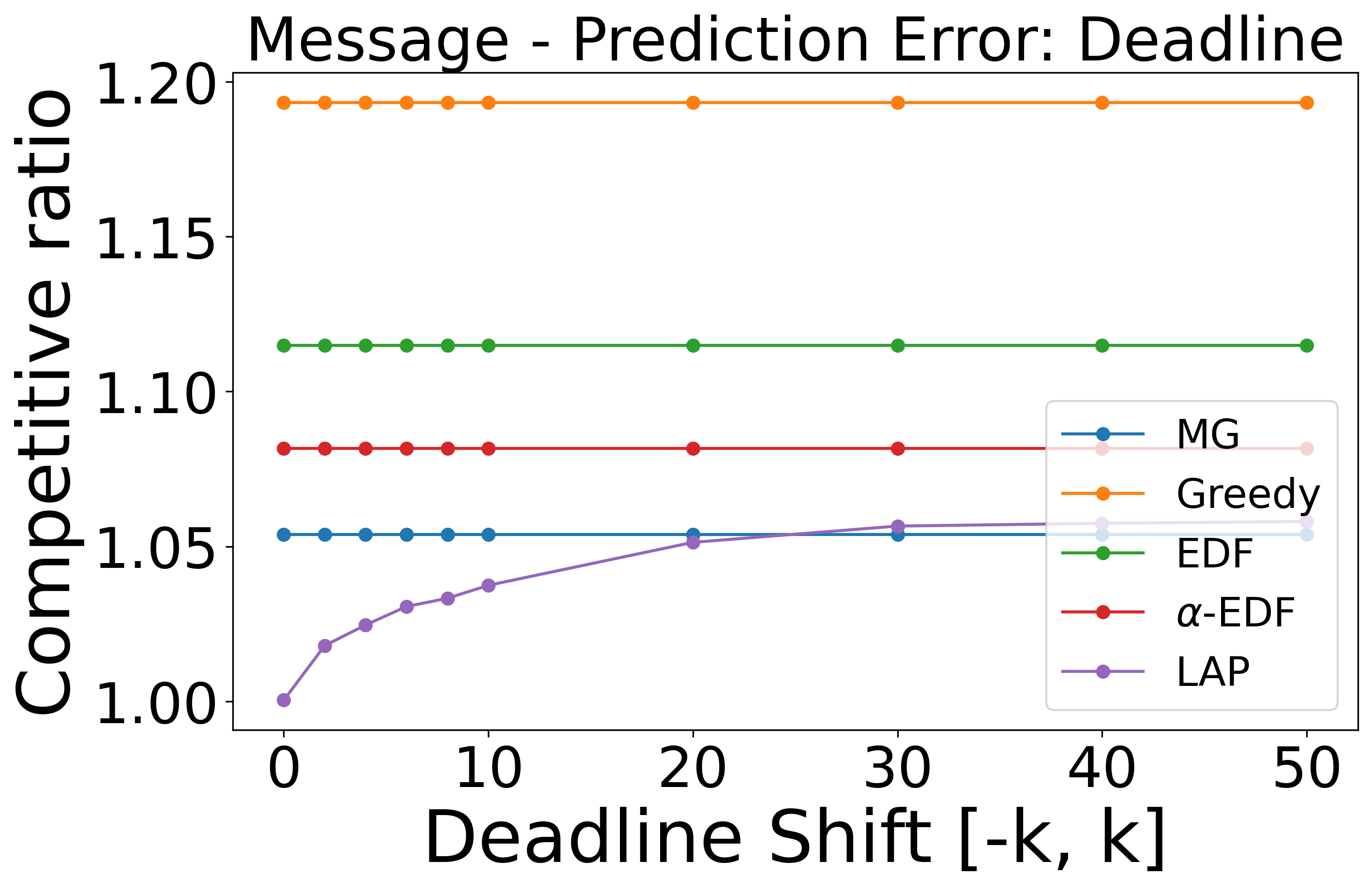}} 
    \subfigure{\includegraphics[width=0.24\textwidth]{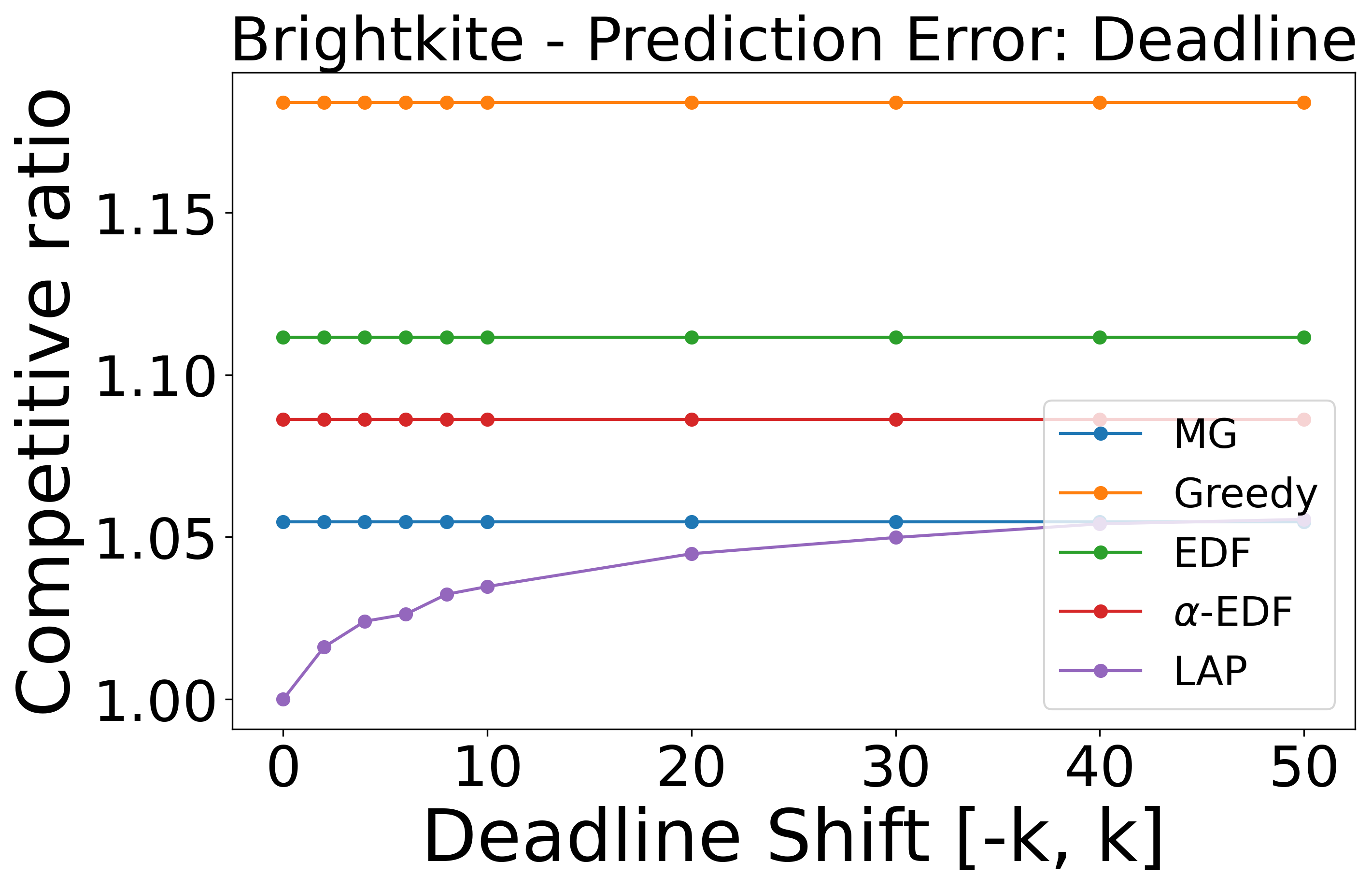}} 
    \caption{\footnotesize{The competitive ratio achieved by our algorithm, \laps, and the benchmark algorithms, as a function of deadline shift $[-k, k]$.}}
    \label{fig:exp:2}
\end{figure*}

For each synthetic dataset, the competitive ratio achieved by the different algorithms is averaged over 10 instances generated i.i.d. For the real datasets, the ratio is averaged over the arrivals for each of the 23 days and 10 days, respectively.

\noindent\textbf{Experiment set 1.} We first evaluate the performance of the algorithms as a function of the error parameter $\sigma$ on packet's weight. In Figure~\ref{fig:exp:1}, we observe that \laps \ outperforms all of the benchmarks when the error parameter is small. Also, the competitive ratio of \laps \ plateaus when the value of $\sigma$ increases, which is consistent with our bounded robustness guarantee.
As previously stated, \textbf{\textsc{EDF}} demonstrates effective performance across various scheduling scenarios. Consequently, we observe that it outperforms other benchmarks without any augmented predictions for synthetic datasets. Furthermore, when arrival times follow a power-law distribution, the performances of benchmarks become more comparable. This is due to one defining feature of the power-law distribution: its heavy tail, often expressed by the 80-20 rule. To be more specific, when the majority of jobs arrive infrequently, with occasional bursts, the differences in the performance of algorithms become less noticeable, except for \textbf{\textsc{Greedy}}. 
However, \textbf{\textsc{EDF}} encounters challenges in real-world scenarios, as illustrated on the right-hand side of Figure~\ref{fig:exp:1}. Given the impracticality of sending a very light earliest-deadline packet when heavy packets are prevalent, \textbf{\textsc{MG}} addresses this difficulty by intelligently identifying packets that merit transmission. It then selectively sends either the earliest-deadline or maximum-weight packet, depending on the ratio of their weights.

\noindent\textbf{Experiment set 2.} In the second set of experiments, we examine the impact of the shift on the packet's deadline. In Figure~\ref{fig:exp:2}, we observe that \laps \ outperforms all of the benchmarks until a large deadline shift. We also observe that the competitive ratio of \laps \ deteriorates as a function of the deadline shift $[-k, k]$, implying that there is a higher possibility of processing the wrong packet when enlarging the shift.

Overall, we observe that \laps\ outperforms all the benchmarks when the error parameter is small, both on packet weight and deadline shift, while still maintaining the competitive ratio when the error is large.

\section{Summary and open problems}
\label{sec:summary}

This paper explores learning-augmented online packet scheduling with deadlines, presenting a deterministic algorithm as its main contribution. The algorithm achieves optimality with accurate prediction and demonstrates improved competitive ratios with small errors, maintaining a constant competitive ratio even with large errors. Additionally, a novel method is introduced to analyze the algorithm by partitioning the time interval into two smaller intervals. A lower-bound tradeoff between consistency and robustness is also proposed.




Lastly, we introduce some open problems.
It would be worthwhile to design a new algorithm or error measure where the competitive ratio is entirely smooth in response to prediction errors. It would also be interesting to narrow the gap between robustness and the 2-competitive lower bound for any 1-consistent algorithm. 
Another direction is to explore the First-In-First-Out (FIFO) model with prediction. In the FIFO model, each packet has no deadline, but the buffer has limited capacity, and each packet must be transmitted in the order of its arrival.



\newpage
\bibliography{reference_deadlines}

\begin{thebibliography}{38}
\providecommand{\natexlab}[1]{#1}
\providecommand{\url}[1]{\texttt{#1}}
\expandafter\ifx\csname urlstyle\endcsname\relax
  \providecommand{\doi}[1]{doi: #1}\else
  \providecommand{\doi}{doi: \begingroup \urlstyle{rm}\Url}\fi

\bibitem[Andelman et~al.(2003)Andelman, Mansour, and Zhu]{andelman2003competitive}
Nir Andelman, Yishay Mansour, and An~Zhu.
\newblock Competitive queueing policies for qos switches.
\newblock In \emph{SODA}, volume~3, pages 761--770. Citeseer, 2003.

\bibitem[Antoniadis et~al.(2020)Antoniadis, Gouleakis, Kleer, and Kolev]{Antoniadis20Secretary}
Antonios Antoniadis, Themis Gouleakis, Pieter Kleer, and Pavel Kolev.
\newblock Secretary and online matching problems with machine learned advice.
\newblock In H.~Larochelle, M.~Ranzato, R.~Hadsell, M.~F. Balcan, and H.~Lin, editors, \emph{Advances in Neural Information Processing Systems}, pages 7933--7944, 2020.

\bibitem[Antoniadis et~al.(2021)Antoniadis, Ganje, and Shahkarami]{antoniadis2021novel}
Antonios Antoniadis, Peyman~Jabbarzade Ganje, and Golnoosh Shahkarami.
\newblock A novel prediction setup for online speed-scaling.
\newblock \emph{arXiv preprint arXiv:2112.03082}, 2021.

\bibitem[Azar et~al.(2022)Azar, Panigrahi, and Touitou]{azar2022online}
Yossi Azar, Debmalya Panigrahi, and Noam Touitou.
\newblock Online graph algorithms with predictions.
\newblock In \emph{Proceedings of the 2022 Annual ACM-SIAM Symposium on Discrete Algorithms (SODA)}, pages 35--66. SIAM, 2022.

\bibitem[Balkanski et~al.(2022)Balkanski, Ou, Stein, and Wei]{balkanski2022scheduling}
Eric Balkanski, Tingting Ou, Clifford Stein, and Hao-Ting Wei.
\newblock Scheduling with speed predictions.
\newblock \emph{arXiv preprint arXiv:2205.01247}, 2022.

\bibitem[Balkanski et~al.(2023)Balkanski, Perivier, Stein, and Wei]{balkanski2023energyefficient}
Eric Balkanski, Noemie Perivier, Clifford Stein, and Hao-Ting Wei.
\newblock Energy-efficient scheduling with predictions.
\newblock In \emph{Thirty-seventh Conference on Neural Information Processing Systems}, 2023.

\bibitem[Bamas et~al.(2020)Bamas, Maggiori, Rohwedder, and Svensson]{bamas2020learning}
{\'E}tienne Bamas, Andreas Maggiori, Lars Rohwedder, and Ola Svensson.
\newblock Learning augmented energy minimization via speed scaling.
\newblock \emph{Advances in Neural Information Processing Systems}, 33:\penalty0 15350--15359, 2020.

\bibitem[Beloglazov and Buyya(2012)]{beloglazov2012managing}
Anton Beloglazov and Rajkumar Buyya.
\newblock Managing overloaded hosts for dynamic consolidation of virtual machines in cloud data centers under quality of service constraints.
\newblock \emph{IEEE transactions on parallel and distributed systems}, 24\penalty0 (7):\penalty0 1366--1379, 2012.

\bibitem[Bernardini et~al.(2022{\natexlab{a}})Bernardini, Lindermayr, Marchetti-Spaccamela, Megow, Stougie, and Sweering]{bernardini2022universal}
Giulia Bernardini, Alexander Lindermayr, Alberto Marchetti-Spaccamela, Nicole Megow, Leen Stougie, and Michelle Sweering.
\newblock A universal error measure for input predictions applied to online graph problems.
\newblock \emph{Advances in Neural Information Processing Systems}, 35:\penalty0 3178--3190, 2022{\natexlab{a}}.

\bibitem[Bernardini et~al.(2022{\natexlab{b}})Bernardini, Lindermayr, Marchetti-Spaccamela, Megow, Stougie, and Sweering]{https://doi.org/10.48550/arxiv.2205.12850}
Giulia Bernardini, Alexander Lindermayr, Alberto Marchetti-Spaccamela, Nicole Megow, Leen Stougie, and Michelle Sweering.
\newblock A universal error measure for input predictions applied to online graph problems, 2022{\natexlab{b}}.
\newblock URL \url{https://arxiv.org/abs/2205.12850}.

\bibitem[Bienkowski et~al.(2011)Bienkowski, Chrobak, and Je{\.z}]{bienkowski2011randomized}
Marcin Bienkowski, Marek Chrobak, and {\L}ukasz Je{\.z}.
\newblock Randomized competitive algorithms for online buffer management in the adaptive adversary model.
\newblock \emph{Theoretical Computer Science}, 412\penalty0 (39):\penalty0 5121--5131, 2011.

\bibitem[Bienkowski et~al.(2013)Bienkowski, Chrobak, D{\"u}rr, Hurand, Je{\.z}, Je{\.z}, and Stachowiak]{bienkowski2013collecting}
Marcin Bienkowski, Marek Chrobak, Christoph D{\"u}rr, Mathilde Hurand, Artur Je{\.z}, {\L}ukasz Je{\.z}, and Grzegorz Stachowiak.
\newblock Collecting weighted items from a dynamic queue.
\newblock \emph{Algorithmica}, 65\penalty0 (1):\penalty0 60--94, 2013.

\bibitem[B{\"o}hm et~al.(2019)B{\"o}hm, Chrobak, Je{\.z}, Li, Sgall, and Vesel{\`y}]{bohm2019online}
Martin B{\"o}hm, Marek Chrobak, {\L}ukasz Je{\.z}, Fei Li, Ji{\v{r}}{\'\i} Sgall, and Pavel Vesel{\`y}.
\newblock Online packet scheduling with bounded delay and lookahead.
\newblock \emph{Theoretical Computer Science}, 776:\penalty0 95--113, 2019.

\bibitem[Chin and Fung(2003)]{chin2003online}
Francis~YL Chin and Stanley~PY Fung.
\newblock Online scheduling with partial job values: Does timesharing or randomization help?
\newblock \emph{Algorithmica (New York)}, 2003.

\bibitem[Chin et~al.(2006)Chin, Chrobak, Fung, Jawor, Sgall, and Tich{\`y}]{chin2006online}
Francis~YL Chin, Marek Chrobak, Stanley~PY Fung, Wojciech Jawor, Ji{\v{r}}{\'\i} Sgall, and Tom{\'a}{\v{s}} Tich{\`y}.
\newblock Online competitive algorithms for maximizing weighted throughput of unit jobs.
\newblock \emph{Journal of Discrete Algorithms}, 4\penalty0 (2):\penalty0 255--276, 2006.

\bibitem[Cho et~al.(2011)Cho, Myers, and Leskovec]{cho2011friendship}
Eunjoon Cho, Seth~A Myers, and Jure Leskovec.
\newblock Friendship and mobility: user movement in location-based social networks.
\newblock In \emph{Proceedings of the 17th ACM SIGKDD international conference on Knowledge discovery and data mining}, pages 1082--1090, 2011.

\bibitem[Etoh and Yoshimura(2005)]{etoh2005advances}
Minoru Etoh and Takeshi Yoshimura.
\newblock Advances in wireless video delivery.
\newblock \emph{Proceedings of the IEEE}, 93\penalty0 (1):\penalty0 111--122, 2005.

\bibitem[Goldwasser(2010)]{goldwasser2010survey}
Michael~H Goldwasser.
\newblock A survey of buffer management policies for packet switches.
\newblock \emph{ACM SIGACT News}, 41\penalty0 (1):\penalty0 100--128, 2010.

\bibitem[Hajek(2001)]{hajek2001competitiveness}
Bruce Hajek.
\newblock On the competitiveness of on-line scheduling of unit-length packets with hard deadlines in slotted time.
\newblock In \emph{Proceedings of the 2001 Conference on Information Sciences and Systems}, 2001.

\bibitem[Je{\.z}(2011)]{jez2011one}
{\L}ukasz Je{\.z}.
\newblock One to rule them all: A general randomized algorithm for buffer management with bounded delay.
\newblock In \emph{ESA}, pages 239--250. Springer, 2011.

\bibitem[Je{\.z} et~al.(2012)Je{\.z}, Li, Sethuraman, and Stein]{jez2012online}
{\L}ukasz Je{\.z}, Fei Li, Jay Sethuraman, and Clifford Stein.
\newblock Online scheduling of packets with agreeable deadlines.
\newblock \emph{ACM Transactions on Algorithms (TALG)}, 9\penalty0 (1):\penalty0 1--11, 2012.

\bibitem[Jin and Ma(2022)]{https://doi.org/10.48550/arxiv.2206.11397}
Billy Jin and Will Ma.
\newblock Online bipartite matching with advice: Tight robustness-consistency tradeoffs for the two-stage model, 2022.
\newblock URL \url{https://arxiv.org/abs/2206.11397}.

\bibitem[Kesselman et~al.(2004)Kesselman, Lotker, Mansour, Patt-Shamir, Schieber, and Sviridenko]{kesselman2004buffer}
Alexander Kesselman, Zvi Lotker, Yishay Mansour, Boaz Patt-Shamir, Baruch Schieber, and Maxim Sviridenko.
\newblock Buffer overflow management in qos switches.
\newblock \emph{SIAM Journal on Computing}, 33\penalty0 (3):\penalty0 563--583, 2004.

\bibitem[Lattanzi et~al.(2020)Lattanzi, Lavastida, Moseley, and Vassilvitskii]{Lattanzi2020online}
Silvio Lattanzi, Thomas Lavastida, Benjamin Moseley, and Sergei Vassilvitskii.
\newblock Online scheduling via learned weights.
\newblock In \emph{Proceedings of the 2020 ACM-SIAM Symposium on Discrete Algorithms (SODA)}, pages 1859--1877, 2020.

\bibitem[Lindermayr and Megow(2022)]{lindermayr2022permutation}
Alexander Lindermayr and Nicole Megow.
\newblock Permutation predictions for non-clairvoyant scheduling.
\newblock In \emph{Proceedings of the 34th ACM Symposium on Parallelism in Algorithms and Architectures}, pages 357--368, 2022.

\bibitem[Lindermayr and Megow(2024)]{algorithms_with_predictions}
Alexander Lindermayr and Nicole Megow.
\newblock Algorithms with predictions.
\newblock \url{https://algorithms-with-predictions.github.io/}, 2024.

\bibitem[Lykouris and Vassilvitskii(2021)]{DBLP:journals/jacm/LykourisV21}
Thodoris Lykouris and Sergei Vassilvitskii.
\newblock Competitive caching with machine learned advice.
\newblock \emph{J. {ACM}}, 68\penalty0 (4):\penalty0 24:1--24:25, 2021.
\newblock \doi{10.1145/3447579}.
\newblock URL \url{https://doi.org/10.1145/3447579}.

\bibitem[Mitzenmacher(2020)]{Mitzenmacher20Scheduling}
Michael Mitzenmacher.
\newblock {Scheduling with Predictions and the Price of Misprediction}.
\newblock In \emph{11th Innovations in Theoretical Computer Science Conference (ITCS 2020)}, volume 151 of \emph{Leibniz International Proceedings in Informatics (LIPIcs)}, pages 14:1--14:18, 2020.
\newblock ISBN 978-3-95977-134-4.

\bibitem[Mitzenmacher and Vassilvitskii(2022)]{MitzenmacherV22}
Michael Mitzenmacher and Sergei Vassilvitskii.
\newblock Algorithms with predictions.
\newblock \emph{Commun. {ACM}}, 65\penalty0 (7):\penalty0 33--35, 2022.
\newblock \doi{10.1145/3528087}.
\newblock URL \url{https://doi.org/10.1145/3528087}.

\bibitem[Panzarasa et~al.(2009)Panzarasa, Opsahl, and Carley]{panzarasa2009patterns}
Pietro Panzarasa, Tore Opsahl, and Kathleen~M Carley.
\newblock Patterns and dynamics of users' behavior and interaction: Network analysis of an online community.
\newblock \emph{Journal of the American Society for Information Science and Technology}, 60\penalty0 (5):\penalty0 911--932, 2009.

\bibitem[Purohit et~al.(2018)Purohit, Svitkina, and Kumar]{Purohit18improving}
Manish Purohit, Zoya Svitkina, and Ravi Kumar.
\newblock Improving online algorithms via ml predictions.
\newblock In S.~Bengio, H.~Wallach, H.~Larochelle, K.~Grauman, N.~Cesa-Bianchi, and R.~Garnett, editors, \emph{Advances in Neural Information Processing Systems}. Curran Associates, Inc., 2018.

\bibitem[Ranjan et~al.(2002)Ranjan, Rolia, Fu, and Knightly]{ranjan2002qos}
Supranamaya Ranjan, Jerome Rolia, Ho~Fu, and Edward Knightly.
\newblock Qos-driven server migration for internet data centers.
\newblock In \emph{IEEE 2002 Tenth IEEE International Workshop on Quality of Service (Cat. No. 02EX564)}, pages 3--12. IEEE, 2002.

\bibitem[Sakr and Stein(2016)]{sakr2016empirical}
Nourhan Sakr and Cliff Stein.
\newblock An empirical study of online packet scheduling algorithms.
\newblock In \emph{Experimental Algorithms: 15th International Symposium, SEA 2016, St. Petersburg, Russia, June 5-8, 2016, Proceedings 15}, pages 278--293, 2016.

\bibitem[Shin et~al.(2001)Shin, Kim, and Kuo]{shin2001quality}
Jitae Shin, Jong~Won Kim, and C-CJ Kuo.
\newblock Quality-of-service mapping mechanism for packet video in differentiated services network.
\newblock \emph{IEEE Transactions on Multimedia}, 3\penalty0 (2):\penalty0 219--231, 2001.

\bibitem[Vesel{\`y}(2021)]{vesely2021packet}
Pavel Vesel{\`y}.
\newblock Packet scheduling: Plans, monotonicity, and the golden ratio.
\newblock \emph{ACM SIGACT News}, 52\penalty0 (2):\penalty0 72--84, 2021.

\bibitem[Vesel{\`y} et~al.(2022)Vesel{\`y}, Chrobak, Je{\.z}, and Sgall]{vesely2022competitive}
Pavel Vesel{\`y}, Marek Chrobak, {\L}ukasz Je{\.z}, and Ji{\v{r}}{\'\i} Sgall.
\newblock A $\phi$-competitive algorithm for scheduling packets with deadlines.
\newblock \emph{SIAM Journal on Computing}, 51\penalty0 (5):\penalty0 1626--1691, 2022.

\bibitem[Xu and Moseley(2022)]{xu2022learning}
Chenyang Xu and Benjamin Moseley.
\newblock Learning-augmented algorithms for online steiner tree.
\newblock In \emph{Proceedings of the AAAI Conference on Artificial Intelligence}, volume~36, pages 8744--8752, 2022.

\bibitem[Zhang et~al.(2004)Zhang, Zhu, and Zhang]{zhang2004end}
Qian Zhang, Wenwu Zhu, and Ya-Qin Zhang.
\newblock End-to-end qos for video delivery over wireless internet.
\newblock \emph{Proceedings of the IEEE}, 93\penalty0 (1):\penalty0 123--134, 2004.

\end{thebibliography}
\bibliographystyle{plainnat}
\newpage
\appendix
\appendix

\section{Missing proofs in \cref{sec:algo_deadlines}}
\label{app:algo}

\lemrobust*
\begin{proof}
For $c > 0$, a parameter that we choose later, we divide the proof into two distinct cases: (1) $W(\OPT(\J)^{(> \tl-1)}) > c W(\OPT(\J)^{(\le \tl-1)})$ and (2) $W(\OPT(\J)^{(> \tl-1)}) \le c W(\OPT(\J)^{(\le \tl-1)})$.
\begin{enumerate}
    \item $W(\OPT(\J)^{(> \tl-1)}) > c W(\OPT(\J)^{(\le \tl-1)}).$ 
    
    Let $\widetilde{\P} = \P^{(\le \tl-1)} \cap \OPT(\J)^{(> \tl-1)}$.
    In this case, we have $\OPT(\J)^{(> \tl-1)} \setminus \widetilde{\P}$ is a feasible candidate for $\P^{(> \tl-1)}$. Since \textsc{OnlineAlg} has a competitive ratio of $\gamma_{\text{on}}$, we have
    \begin{equation}\label{equ:2}
        \gamma_{\text{on}}W(\P^{(> \tl-1)}) \ge W(\OPT(\J)^{(> \tl-1)}) - W(\widetilde{\P}) \ .
    \end{equation}
    
    By applying the condition $W(\OPT(\J)^{(> \tl-1)}) > c W(\OPT(\J)^{(\le \tl-1)})$, we can obtain 
    \begin{align*}
        \frac{c+1}{c}\gamma_{\text{on}}W(\P) &= \frac{c+1}{c}\gamma_{\text{on}} (W(\P^{(\le \tl-1)})+ W(\P^{(> \tl-1)})) \\
        &\ge \frac{c+1}{c}\gamma_{\text{on}} W(\P^{(\le \tl-1)}) +  \frac{c+1}{c} (W(\OPT(\J)^{(> \tl-1)}) - W(\widetilde{\P})) \\
        & \ge \frac{c+1}{c}W(\OPT(\J)^{(> \tl-1)})\\
        &> W(\OPT(\J)^{(\le \tl-1)}) + W(\OPT(\J)^{(> \tl-1)})\\
        & = W(\OPT(\J))\ ,
    \end{align*}
    where the first inequality uses \eqref{equ:2}, the second inequality uses the fact that $\widetilde{\P} \subseteq \P^{(\le \tl-1)}$, and the last inequality uses the assumption of this case, where $W(\OPT(\J)^{(> \tl-1)}) > c W(\OPT(\J)^{(\le \tl-1)})$.

    \item $W(\OPT(\J)^{(> \tl-1)}) \le c W(\OPT(\J)^{(\le \tl-1)}).$ For this case, we further consider the following two cases:
    
    \begin{itemize}
        
        \item $W(\OPT(\J)^{(\le \tl-1)}) = W(\OPT(\J_{\le \tl-1})^{(\le \tl-1)}).$
        
        In this case, since \laps \ switches to \textsc{OnlineAlg} at time $\tl$, it satisfies that $\rho W(\P^{(\le \tl-1)}) \ge W(\OPT(\J_{\le \tl-1})^{(\le \tl-1)})$. Therefore, we can obtain the following result:
        
        \begin{align*}
            (c + 1)\rho  W(\P) 
            & \ge (c + 1) W(\OPT(\J_{\le \tl-1})^{(\le \tl-1)}) \\
            & = (c + 1) W(\OPT(\J)^{(\le \tl-1)}) \\
            & \ge W(\OPT(\J)^{(\le \tl-1)}) + W(\OPT(\J)^{(> \tl-1)})\\
            & = W(\OPT(\J))\ ,
        \end{align*}
        
        where the third inequality applies the condition $W(\OPT(\J)^{(> \tl-1)}) \le c W(\OPT(\J)^{(\le \tl-1)})$.
        
        \item $W(\OPT(\J)^{(\le \tl-1)}) > W(\OPT(\J_{\le \tl-1})^{(\le \tl-1)})$. 
        
        In this case, by Lemma~\ref{lem:maximize}, we have $\OPT(\J)^{(\le \tl-1)} \cap \OPT(\J_{\le \tl-1})^{(> \tl-1)} \neq \emptyset$. Let $\Tilde{\P} = \OPT(\J)^{(\le \tl-1)} \cap \OPT(\J_{\le \tl-1})^{(> \tl-1)}$. We first claim the following property:
        
        \begin{claim}\label{claim:local}
        $W(\OPT(\J_{\le \tl-1})^{(\le \tl-1)}) \ge W(\OPT(\J)^{(\le \tl-1)}) - W(\Tilde{\P})$.
        \end{claim}
        
        \begin{proof}
            Assume $W(\OPT(\J_{\le \tl-1})^{(\le \tl-1)}) < W(\OPT(\J)^{(\le \tl-1)}) - W(\Tilde{\P})$. By replacing $\OPT(\J_{\le \tl-1})^{(\le \tl-1)}$ with  $\OPT(\J)^{(\le \tl-1)} \setminus \Tilde{\P}$, we have
            
            \begin{align*}
                & W(\OPT(\J)^{(\le \tl-1)} ) - W (\Tilde{\P}) +  W(\OPT(\J_{\le \tl-1})^{(> \tl-1)})\\ 
                & > W(\OPT(\J_{\le \tl-1})^{(\le \tl-1)}) + W(\OPT(\J_{\le \tl-1})^{(> \tl-1)})\\
                & > W(\OPT(\J_{\le \tl-1})) \ ,
            \end{align*} 
            
            which contradicts the optimality of $\OPT(\J_{\le \tl-1})$.
        \end{proof}
        
        Also, since \laps \ switches to \textsc{OnlineAlg} at time $\tl$, we have (i) $\rho W(\P^{(\le \tl-1)}) \ge W(\OPT(\J_{\le \tl-1})^{(\le \tl-1)})$. 
        Moreover, we have (ii) each job $j \in \Tilde{\P}$ is either in $\P^{(\le \tl-1)}$ or feasible for $\P^{(> \tl-1)}$. Specifically, we let $\Tilde{\P}_1 = \Tilde{\P} \cap \P^{(\le \tl-1)}$ and $\Tilde{\P}_2 = \Tilde{\P} \setminus \Tilde{\P}_1$. With the above definition, we next decompose $\P^{(\le \tl-1)}$ into $\P^{(\le \tl-1)^{-}}$ and $\Tilde{\P}_1$. By (i), we can obtain: 
        
        \begin{align*}
            \rho W(\P^{(\le \tl-1)}) &= \rho(W(\P^{(\le \tl-1)^{-}}) + W(\Tilde{\P}_1))\\ 
            &\ge W(\OPT(\J_{\le \tl-1})^{(\le \tl-1)}) \\ 
            & \ge W(\OPT(\J)^{(\le \tl-1)}) - W(\Tilde{\P}) \ ,
        \end{align*}
        
        where the first inequality follows from (i) and the second inequality follows from Claim~\ref{claim:local}. 
        
        Next, since \textsc{OnlineAlg} is $\gamma_{\texttt{on}}$-competitive, we have 
        
        \begin{equation}\label{equ:3}
            \gamma_{\texttt{on}}W(\P^{(> \tl-1)}) \ge W(\Tilde{\P}_2)\ .
        \end{equation}
        
        Finally, by combining the above two inequalities, and using the fact that $\rho + 1 > \gamma_{\texttt{on}}$, we can obtain
        
        \begin{align*}
            (\rho + 1) W(\P) & = \max\{\rho + 1, \gamma_{\texttt{on}}\}W(\P)\\ &\ge (\rho + 1) W(\P^{(\le \tl-1)}) + \gamma_{\texttt{on}}W(\P^{(> \tl-1)})\\ &\ge W(\OPT(\J)^{(\le \tl-1)}) - W(\Tilde{\P}) + W(\Tilde{\P}_1) + W(\Tilde{\P}_2)\\
            & \ge W(\OPT(\J)^{(\le \tl-1)}) \ ,
        \end{align*}
        
        where the second inequality uses the result $\rho W(\P^{(\le \tl-1)}) \ge W(\OPT(\J)^{(\le \tl-1)}) - W(\Tilde{\P})$ and \eqref{equ:3}, and the last inequality uses the fact that $\Tilde{\P} = \Tilde{\P}_1 \cup \Tilde{\P}_2$.
        
        Finally, by using the assumption $W(\OPT(\J)^{(> \tl-1)}) \le c W(\OPT(\J)^{(\le \tl-1)})$, we have the competitive ratio of this case is $(c+1)(\rho + 1)$.
    \end{itemize}
    
    By combining the above three cases, and we can derive the competitive ratio as $$\max\{\frac{\gamma_{\texttt{on}}(c+1)}{c}, (c+1)(\rho + 1)\}.$$ To optimize the competitive ratio, we set $c = \frac{\gamma_{\texttt{on}}}{\rho+1}$. Substituting this value into the equation yields a competitive ratio of $\gamma_{\texttt{on}} + \rho + 1$. This completes the proof.
\end{enumerate}

\end{proof}

\section{Missing proofs in \cref{sec:lower-bound}}
\label{app:lb}

\proplb*

\begin{proof}
Consider the following two problem instances: $\J_1$ and $\J_2$, where $\J_1 = \{(0, 1, \epsilon), (0, 2, 1)\}$ and $\J_2 = \{(0, 1, \epsilon), (0, 2, 1), (1, 2, 1)\}$, respectively.

Assume that, initially, a 1-consistent algorithm $\A$ is given $\J_1$ as the prediction. We first argue that, for any 1-consistent algorithm, it will process the job with weight $\epsilon$ at time $t = 0$; otherwise, if such an algorithm does not process it, the job will expire at time $t = 1$, leading to a contradiction.

After observing the decision made by algorithm $\A$, the adversary will choose the problem instance $\J_2$. Therefore, at time $t = 1$, algorithm $\A$ realizes that there is another arrival with the same deadline at $t =2$ and weight $1$, and the algorithm can only process one of the two.

On the other hand, the optimal solution for $\J_2$ would first process the job with weight $1$ at time $t = 0$, resulting in a total processing weight of $2$. Therefore, the competitive ratio would be $\frac{2}{1+\epsilon}$.
By letting the value of $\epsilon$ be very small, the ratio becomes $2$, and we conclude the proof.

\end{proof}

\thmlb*

\begin{proof}
Consider the following two problem instances: $\J_1$ and $\J_2$, where $\J_1 = \{(0, 1, \alpha + \epsilon), (0, 2, 1)\}$ and $\J_2 = \{(0, 1, \alpha + \epsilon), (0, 2, 1), (1, 2, 1)\}$, respectively.

Assume that, initially, a $1+\alpha$-consistent algorithm $\A$ is given $\J_1$ as the prediction. We first argue that, for any $1+\alpha$-consistent, it will process the job with weight $\alpha + \epsilon$ at time $t = 0$. Otherwise, if such an algorithm does not process this job, it will expire at time $t = 1$, resulting in a consistency value of $\frac{1 + \alpha + \epsilon}{1} > 1 + \alpha$, which leads to a contradiction.

Upon observing the decision made by algorithm $\A$ regarding the job with weight $\alpha + \epsilon$, the adversary will select problem instance $\J_2$. At time $t = 1$, the adversary releases another job, $j = (1, 2, 1)$. Consequently, at time $t = 1$, algorithm $\A$ realizes that there is another arrival with the same deadline and weight, and the algorithm can only process one of the two.

On the other hand, the optimal solution for $\J_2$ processes the job with weight $1$ at time $t = 0$, resulting in a total processing weight of $2$. Therefore, the competitive ratio becomes $\frac{2}{1 + \alpha} $ by letting the value of $\epsilon$ be very small, concluding the proof.
\end{proof}

\end{document}